\documentclass[11pt]{article}
\usepackage{amsmath,amssymb, amsfonts}
\usepackage{epsfig, fullpage, graphics}

\advance \topmargin by -\headheight
\advance \topmargin by -\headsep
     
\evensidemargin \oddsidemargin
\newtheorem{theorem}{Theorem}
\newtheorem{lemma}[theorem]{Lemma}
\newtheorem{proposition}[theorem]{Proposition}
\newtheorem{claim}[theorem]{Claim}
\newtheorem{corollary}[theorem]{Corollary}
\newtheorem{definition}{Definition}
\newenvironment{proof}{{\bf Proof:}}{\hfill\rule{2mm}{2mm}}
\newenvironment{remark}{{\bf Remark:}}{}

\newcommand{\comment}[1]{}
\DeclareMathOperator\Ex{E}
\DeclareMathOperator\Ind{Ind}
\newcommand{\sat}{{\sc Sat}}
\newcommand{\unc}{\ensuremath{o}}
\newcommand{\Weight}{\ensuremath{W}}
\newcommand{\mb}[1]{\ensuremath{\mathbf{#1}}}
\newcommand{\xvec}{\ensuremath{\mb{x}}}
\newcommand{\svec}{\ensuremath{\mb{\sigma}}}
\newcommand{\tvec}{\ensuremath{\mb{\tau}}}
\newcommand{\defn}{\ensuremath{:  =}}
\newcommand{\sato}[2]{\ensuremath{s_{#1,#2}}}
\newcommand{\unsat}[2]{\ensuremath{u_{#1,#2}}}

\title{On the satisfiability threshold and clustering of solutions 
\\of random 3-\sat~formulas}
\author{Elitza Maneva\thanks{IBM Almaden Research Center, San Jose, CA, 
enmaneva@us.ibm.com. Some of this research was done while this author
was a PhD student at UC Berkeley, supported in part by NSF grant
DMS-0528488.}
\and Alistair Sinclair\thanks{Computer Science Division, University of California at Berkeley,
sinclair@cs.berkeley.edu.  Supported in part by NSF grants DMS-0528488
and CCR-0635153.}}
\date{}

\begin{document}

\maketitle

\begin{abstract}
We study the structure of satisfying assignments of a random 3-\sat\
formula. In particular, we show that a random formula of density
$\alpha \ge 4.453$ almost surely has no non-trivial ``core''
assignments. Core assignments are certain partial assignments that can
be extended to satisfying assignments, and have been studied recently
in connection with the Survey Propagation heuristic for random \sat.
Their existence implies the presence of clusters of solutions, and
they have been shown to exist with high probability below the
satisfiability threshold for $k$-\sat\ with $k\ge 9$ \cite{AR06}.  Our
result implies that either this does not hold for 3-\sat\ or the
threshold density for satisfiability in 3-\sat\ lies below $4.453$.
The main technical tool that we use is a novel simple application of the 
first moment method.

\end{abstract}

\section{Introduction}

The study of random instances of 3-\sat\ has been a major research
focus in recent years, both because of its inherent interest and
because it is a natural test case for the wider understanding of the
complexity of computational tasks on random inputs.  In random 3-\sat\
the input is a formula drawn uniformly at random from all formulas of
fixed density~$\alpha$, i.e., formulas with $\alpha n$ 
clauses on $n$ variables.  Friedgut \cite{Friedgut99} proved that
there exists a function $\alpha_c(n)$, known as the {\it
satisfiability threshold}, such that for any positive $\epsilon$,
random formulas of density $\alpha_c(n)-\epsilon$ have satisfying
assignments with high probability, and random formulas of density
$\alpha_c(n)+\epsilon$ have no satisfying assignment with high
probability. It is conjectured that $\alpha_c(n)$ is a constant (and
that its value is about 4.27), but currently all that is known is
that, for large~$n$, $3.520 \le
\alpha_c(n) \le 4.506$ \cite{KKL06, HS03, DBM00}.

In the range of densities for which the formula is satisfiable with
high probability, the interesting algorithmic question is whether we
can find even one of the many satisfying assignments in polynomial
time. The lower bound on $\alpha_c(n)$ is a result of the analysis of
such a polynomial time algorithm~\cite{KKL06, HS03}. This algorithm
belongs to a family of algorithms known as ``myopic'' because they
assign variables greedily one by one in an order that is based only on
the number of positive and negative occurrences of each variable.

An apparently much more powerful algorithm is Survey
Propagation~\cite{MPZ02, MZ02}. In experiments on very large instances
(say, with $n=10^6$ variables) it finds solutions for formulas of
densities only just below the conjectured threshold value
$\alpha=4.27$; however, a rigorous analysis of its performance is
still far from our reach. Like the myopic algorithms it also assigns
variables one by one in a greedy manner, but its choices are based on
more global information about the role of each variable in the
formula. That information is provided by the fixed point of a
sophisticated message passing dynamics between variables and clauses.

The message passing procedure is based on an intriguing (informal)
picture of the properties of the solution space of a random formula,
which is derived from the 1-step Replica Symmetry Breaking ansatz of
statistical physics~\cite{MPV87}. A key postulate of this physical
picture is that, for formulas of density higher than a certain value
(estimated to be $3.92$), the space of solutions is split into
``clusters.''  Within the same cluster it is possible to reach any
satisfying assignment from any other by flipping one variable at a
time, while always keeping the formula satisfied. On the other hand,
in order to get from a solution in one cluster to a solution in
another, the values of a linear number of variables have to be flipped
at the same time. Loosely speaking, the message passing procedure of 
Survey Propagation was proposed as a way to collect information 
about the {\it clusters\/} of assignments, rather than individual assignments.

The remarkable performance of Survey Propagation is a compelling
reason to explore properties of the solution space of typical formulas
as a way to further our understanding of random 3-\sat, its hardness,
and the satisfiability threshold, and also (looking much further ahead) in
order to systematically design algorithms for more general
problems with distributional inputs.

A detailed study of the Survey Propagation algorithm undertaken
in~\cite{MMW07} and \cite{BZ04} led to an interpretation of the
message passing procedure as the more familiar Belief Propagation
algorithm~\cite{Pearl88} applied to a particular probability
distribution on {\it partial\/} assignments, i.e., assignments of
values from the set $\{0,1,\ast\}$. Here $\ast$ is to be interpreted
as ``unassigned,'' and a variable is allowed to be unassigned only if
it is not forced to be assigned 0 or 1 in order to satisfy a clause.

Among these partial assignments, the set of ``core'' assignments plays a
central role.  These are partial assignments that are obtained 
from a satisfying assignment by successively replacing each
unconstrained variable by~$\ast$.  (A variable is {\it unconstrained\/}
if changing its value does not make any clause unsatisfied.)
Any satisfying assignment has a unique corresponding core.  Moreover,
since all assignments in a cluster have the same core, cores can be
viewed (informally) as ``summaries'' of clusters.  Of course, this view is
useful only if different clusters tend to have distinct non-trivial
cores. (The trivial core assignment is the one without any assigned
variables.)  Recently, Achlioptas and Ricci-Tersenghi~\cite{AR06}
showed that, in random $k$-\sat\ for $k\ge 9$, for some range of
densities up to the satisfiability threshold, with high probability
{\it every satisfying assignment has a non-trivial core assignment
associated with it}. This implies that clusters have a large number of
frozen variables; indeed, for large $k$ the fraction of frozen
variables comes arbitrarily close to 1. (The clustering picture has
also been confirmed for $8$-\sat\ by
\cite{MMZ05} and~\cite{AR06} using a different method that does not
say anything about cores.)

The above results hold only for random $k$-\sat\ with $k\ge 8$ or $9$.  In
this paper we investigate similar questions for the apparently harder case
of random 3-\sat.  Our main result is the following:
\begin{theorem}
\label{thm:main}
For random instances of 3-\sat\ with density greater than $4.453$, with
high probability there exist \emph{no} non-trivial core assignments.
\end{theorem}

This theorem requires some interpretation.  Note first that the density
$4.453$ lies above the conjectured threshold value of $4.27$ but
below the current best known upper bound of~$4.506$.  Thus we may deduce:
\begin{corollary}
\label{cor} One of the following statements holds for random 3-\sat:
\begin{itemize}
\item $\alpha_c(n) \le 4.453$; or
\item there is a range of densities immediately below the
satisfiability threshold for which with high probability there are no
non-trivial core assignments.
\end{itemize}
\end{corollary}
One interpretation of Theorem~\ref{thm:main} is as evidence for
an improved upper bound on the threshold of the form $\alpha_c(n) \le
4.453$.  On the other hand, if in fact $\alpha_c(n) > 4.453$ then the
theorem establishes a range of densities immediately below the
threshold for which with high probability the formula is satisfiable
but has no non-trivial core assignments.  This would 
represent a surprising difference between the properties of random
9-\sat\ and random 3-\sat. Interestingly, experiments with 3-\sat\ and
solutions of large formulas found by Survey Propagation do not find
cores (see \cite{MMW07}). It is quite conceivable that both
of the statements in the above corollary hold.

We now briefly discuss our proof technique, which involves a novel
application of the first moment method and is, we believe, perhaps
more noteworthy than the result itself.  A straightforward application of the first
moment method to the set of core assignments only allows us to show
that with high probability there are no cores of small size or of
large size. To handle core assignments of intermediate size, it is further
necessary to bound the probability that a partial assignment can be
extended to a satisfying assignment. This probability is equivalent to
the probability that a random formula with a given density of 2-clauses
and 3-clauses is satisfiable. For this sub-problem, it is natural to use
one of the methods that have been previously introduced for
bounding the probability of satisfiability of random 3-\sat\ formulas
\cite{KMPS95,DB97,KKKS98,JSV00,KKSVZ07,DBM00}. However, 
the most powerful method, from~\cite{DBM00}, is very heavy numerically,
and it does not seem possible to carry it out for the whole range of
required densities; on the other hand, simpler methods such as those
in~\cite{KKKS98} are apparently not powerful enough. 
Instead we introduce a new method, which we now briefly outline
(for a detailed description, see Section~\ref{sec:cores}).

Traditionally, bounds on the probability of satisfiability such as those
mentioned above are obtained by
applying the first moment method to a random variable which counts the
number of satisfying assignments of a particular kind. Indeed, much of the
work on bounding the satisfiability threshold has been directed
towards identifying a set of satisfying assignments that is a strict
subset of the set of all satisfying assignments (so that its
expected size is much smaller), but is always non-empty if the formula is
satisfiable. The novelty in our approach lies in identifying a new random
variable which depends not only on satisfying assignments but also on
partial assignments. This random variable is at least~1 for every
satisfiable formula, it exploits the clustering structure of the solution space,
and most importantly, it is significantly easier to compute than many
alternative approaches.

Finally we note that the proof of our theorem depends on a claim that a 
particular analytic function takes only negative values in a given range. 
We do not provide a complete proof of this numerical claim, but we outline
the steps needed for completing the proof, and give very strong
numerical evidence that all the corresponding statements hold. The
difficulty with making the proof completely rigorous is that this would
require too much computational effort, which we feel is not justified 
at this stage given that the bound we obtain is still some distance from
the satisfiability threshold. The results presented here should rather be
considered as a proof of concept.

The remainder of this paper is structured as follows.  In Section~2 we
give necessary background and precise definitions of the various
concepts used in the paper.  Section~3 is devoted to the proof of our
main result, Theorem~\ref{thm:main}.  We conclude in Section~4 with 
some final remarks and suggestions for future work.

\section{Technical definitions}

Let $x_1, \dots, x_n$ be a set of $n$ Boolean variables. A literal is
either a variable or its negation. A 3-\sat\ formula is a formula in
CNF, where each clause is a disjunction of 3 distinct literals (on
different variables).  For every clause $c$ we will denote the set of
variables that appear in the clause by $V(c)$. The distribution that
we consider is the following: for a given density $\alpha$ we choose
uniformly at random and with replacement $m=\lfloor \alpha n \rfloor$
clauses out of all possible clauses on 3 distinct variables. For a
clause $c$ and a variable $x_i\in V(c)$ we denote by $\sato{c}{i}$ and
$\unsat{c}{i}$ the value for variable $x_i$ that is respectively
satisfying and unsatisfying for clause $c$.

Suppose that the variables $\xvec = (x_1, \ldots, x_n)$ are allowed
to take values in $\{0, 1, \ast\}$, which we refer to as a \emph{partial
assignment}. A variable taking value~$\ast$ (star) should be thought
of as unassigned.

\begin{definition}{\rm
A partial assignment to $\xvec$ is \emph{invalid} for a clause $c$
if either (a) all variables are unsatisfying; or
(b) all variables are unsatisfying except for one index 
$j \in V(c)$, for which $x_j=\ast$.
Otherwise, the partial assignment is \emph{valid} for clause
$c$. We say that a partial assignment is valid for a
formula (or just ``valid'') if it is valid for all of its clauses.}
\end{definition}

For a valid partial assignment, the subset of variables that are
assigned either 0 or 1 values can be divided into \emph{constrained}
and \emph{unconstrained} variables in the following way:
\begin{definition}{\rm
We say that a variable $x_i$ is the \emph{unique satisfying variable}
for a clause $c$ if it is assigned $\sato{c}{i}$ whereas all other
variables in the clause (i.e., the variables $\{ x_j \; : \; j \in
V(c)\backslash \{i\} \}$) are assigned $\unsat{c}{j}$.  A
variable $x_i$ is \emph{constrained} by clause $c$ if it is the unique
satisfying variable for~$c$.}
\end{definition}
A variable is \emph{unconstrained} if it has 0 or 1 value, and is not
constrained by any clause. Thus for any partial assignment $\svec \in
\{0,1,\ast\}^n$ the variables are divided into stars, constrained and
unconstrained variables. Let $S^\ast(\svec)$ be the set of unassigned
variables, and $n_\ast(\svec)$ and $n_\unc(\svec)$ denote respectively
the number of stars and the number of unconstrained variables.  We
define the {\it weight\/} of a valid partial assignment to be
\begin{equation}
\label{eq:weight}
\Weight(\svec) \defn \rho^{n_{\ast}(\svec)} (1-\rho)^{n_{\unc}(\svec)},
\end{equation}
where $\rho$ is a parameter in the interval $[0,1]$. The weight of an invalid
partial assignment is~$0$. 

In \cite{MMW07} the Survey Propagation algorithm is interpreted as a special case of a larger family of Belief Propagation algorithms applied to
a family of distributions on valid partial assignments. This family of 
distributions, parameterized by $\rho$, is defined as $\Pr[\svec]\propto \Weight(\svec)$.
At one extreme, $\rho=0$, this becomes just the uniform distribution over
(full) satisfying assignments.  The other extreme, $\rho=1$, 
corresponds to Survey Propagation.
The pure version of Survey Propagation corresponds to setting $\rho=1$. 
Intermediate values of~$\rho$ interpolate between these extremes.

Next we define a natural partial order (represented by an acyclic
directed graph) on valid partial assignments. The vertex set of the
directed graph $G$ consists of all valid partial assignments.  The
edge set is defined in the following way: for a given pair of valid
partial assignments $\svec$ and $\tvec$, the graph includes a directed
edge from $\svec$ to $\tvec$ if there exists an index $i \in \{1,
\dots, n\}$ such that (i) $\sigma_j = \tau_j$ for all $j \neq i$; and
(ii) $\tau_i = \ast$ and $\sigma_i \neq \ast$.

Valid partial assignments can be separated into levels based on their
number of star variables, i.e., the assignment $\svec$ is in level
$n_\ast(\svec)$. Thus every edge goes from an assignment in level
$l-1$ to one in level $l$, where $1\le l\le n$.  $G$ is acyclic and we
write $\tvec< \svec$ if there is a directed path in $G$ from $\svec$
to $\tvec$. In this case we will also say that the assignment $\tvec$ is 
\emph{consistent} with the assignment $\svec$.
%
The outgoing edges of any valid partial assignment $\svec$ correspond
to its unconstrained variables, and therefore its outdegree is equal
to $n_\unc(\svec)$.
The minimal assignments in this ordering are the assignments without
unconstrained variables, i.e., the positive weight assignments for
$\rho=1$.

\begin{definition}{\rm A {\em core} of a satisfying assignment $\svec$ is a
minimal assignment $\tvec$ such that $\tvec < \svec$.}
\end{definition}

The following proposition about cores is proved in \cite{MMW07}. 
\begin{proposition}\cite{MMW07} 
Any satisfying assignment $\svec$ has a unique core.
Furthermore, if satisfying assignments $\svec^1, \svec^2 \in
\{0,1\}^n$ belong to the same cluster of solutions then they have the
same core.
\end{proposition}

In the above proposition a cluster is simply a connected component of
the graph on solutions, in which two solutions are connected by an
edge if and only if they are at Hamming distance 1.

\begin{definition}{\rm A \emph{cover} is a valid partial assignment that contains no unconstrained variables. }
\end{definition}
In particular, the core of any satisfying assignment is a cover. On
the other hand not all cover assignments are cores (because they may
not be extendable to satisfying assignments). We say that a cover
assignment~$\tvec$ is \emph{non-trivial} if $n_\ast(\tvec) < n$, so
that it has at least one assigned variable.


The proof of Theorem~\ref{thm:main} uses a surprising property of the
weights~(\ref{eq:weight}), which was observed in~\cite{MMW07} but was
not utilized there. (A more general statement and a connection to
a combinatorial object known as ``convex geometry'' was developed in
\cite{AM07}.) Specifically, the total weight of partial assignments
consistent with a given satisfying assignment is exactly 1. This fact
implies that the probability of satisfiability is at most the expected
total weight of partial assignments.

\begin{theorem} \cite{MMW07} 
\label{thm:sum}
For every $\rho\in [0,1]$, $\sum_{\tvec \leq \svec} \Weight(\tvec) =
\rho^{n_\ast(\svec)}$ for any valid partial assignment 
$\svec \in \{0,1,\ast\}^n$. In particular, if $\svec$ is a satisfying
assignment then $\sum_{\tvec \leq \svec} \Weight(\tvec) = 1$.
\end{theorem}

In our proof, we will apply Theorem~\ref{thm:sum} with various values
for $\rho \in (0.8, 1)$. In each application we will chose the value
of~$\rho$ to get the best bound possible for the probability that a
formula chosen from some distribution has a satisfying assignment. In
particular, if a satisfying assignment exists, the total weight of
valid partial assignments is at least~1.  Therefore, the probability
of satisfiability is at most as large as the expected value of the
total weight of partial assignments. Applying this idea directly to
the original distribution on 3-\sat\ formulas leads to an upper bound
on the threshold~$\alpha_c$ that is weaker than the currently best
known bound of~$4.506$.  However, the derivation is simpler than other
approaches, which makes it possible to apply the same method to bound
the probability that a fixed valid partial assignment can be extended
to a satisfying assignment. (This amounts to applying the method to
random formulas coming from a variety of distributions on formulas
with both 2-clauses and 3-clauses.)  This allows us to estimate the
probability of the existence of a non-trivial core, and thus to prove
the theorem.

\section{Proof of the main theorem}
\label{sec:proof}

This section contains a proof of Theorem~\ref{thm:main}.  We begin
with an overview of the entire proof, in the course of which we will state
various technical lemmas; these lemmas will be proved in the 
three subsections that follow.

Note first that for $\alpha\ge 4.506$ the statement of the main theorem 
follows from the fact that random 3-\sat~formulas of density at least $4.506$ 
are known to be unsatisfiable with high probability~\cite{DBM00}.  Hence
from now on we focus on the case that $\alpha \in [4.453, 4.506]$.  Our
goal is to prove that, for densities above $4.453$, with high probability there are
no non-trivial covers that can be extended to satisfying assignments,
i.e., there are no non-trivial cores.  

To this end, define the {\em size} of a cover (or
of a core) to be the number of variables assigned value~0 or~1. The
following lemma,
proved in~\cite{MMW07}, establishes that with high probability all
non-trivial covers (and consequently cores) are of linear size.

\begin{lemma} \cite{MMW07}
\label{lem:epsilon}
For a random $3$-\sat\ formula of density $\alpha$, with high
probability there are no non-trivial covers of size strictly less than
$\frac{1}{\alpha e^2}n$.
\end{lemma}
This lemma implies that it is sufficient to consider core
assignments of size $an$ for $a\in[1/(\alpha e^2),1]$. Let $X_a$
denote the number of cover assignments of size $\lfloor an \rfloor$,
and $Y_a$ denote the number of core assignments of size $\lfloor an
\rfloor$.

In addition to the density~$\alpha$, the size measure~$a$, and the
weight parameter~$\rho\in[0,1]$ from Theorem~\ref{thm:sum}, we will
use two further parameters, $d$ and~$b$, whose precise definitions
will be given in the proofs in the following subsections.  Roughly
speaking, $d$ is the fraction of clauses that are constraining with
respect to a given partial assignment, and $b$ is the fraction of
constrained variables with respect to a given partial assignment. The
ranges of these parameters are $d \in (a/\alpha,1]$ and $b\in[0,1-a]$.
We will also make use of a derived parameter~$r$ which is defined by
$d$, $a$ and $\alpha$. It appears in connection to the event that the
constraining clauses succeed in constraining all constrained
variables. Specifically, $r$ is the value satisfying the equation
$d=\frac{ar}{\alpha}\ln\frac{r}{r-1}$.  (Note that such a value~$r$
always exists and is unique as the right-hand side is monotonic
in~$r$.)  In fact, because of the form of this expression, it will be
more convenient to think of $d$ as being determined by $r$, $a$, and
$\alpha$.  Whenever we take the supremum over one of these parameters,
we always mean the supremum over its allowed range.

We now define two functions that play a central role in our analysis:
\begin{eqnarray*}
f(\alpha, a, r) &:=& a \ln(2) +H(a)+\alpha H(d)+
\alpha d \ln\left(\frac{3a^3}{8}\right) +
\alpha(1-d) \ln\left(1-\frac{a^2(3-a)}{4}\right) \\
&& {}+\alpha d\ln(r/e) - a \ln(r-1);  \\ \\
h(\alpha, a, r, \rho, b) 
&:=& b \ln(2) + (1-a-b)\ln(\rho) +(1-a) H(b/(1-a))\\
 &&{}- \alpha(1-d)b~ \frac{b(6-5b-15a) + 12(1-a)a}{2(4 - a^2(3-a))} 
 + b \ln\left(1-\rho
e^{ -3\alpha(1-d)\frac{b(b+2a)}{2(4-a^2(3-a))}}\right).
\end{eqnarray*}
Here $H$ denotes the entropy function $H(x)=-x\ln(x)-(1-x)\ln(1-x)$.

The first ingredient in the proof of Theorem~\ref{thm:main} is the following
lemma, whose proof is presented in Section~\ref{sec:covers}.
\begin{lemma}  
\label{lem:covers}
For a random $3$-\sat\ formula of density $\alpha\ge 1$, and for every 
$a\in[0,1]$,
$$ \lim_{n \rightarrow \infty}\frac{1}{n} \ln\left(\Ex[X_a]\right)\le
\sup_r f(\alpha, a,r).$$
\end{lemma}
A simple application of Markov's inequality immediately yields:
\begin{corollary}
\label{cor:covers} 
If $\alpha\ge 1$ and $a\in[0,1]$ are such that, for every $r>1$ 
with $d=\frac{ar}{\alpha} \ln\frac{r}{r-1}\le 1$, it holds that $f(\alpha,
a, r)<0$, then with high probability random 3-\sat\ formulas of
density $\alpha$ do not have covers of size $an$.
\end{corollary}

The second ingredient in the proof of Theorem~\ref{thm:main} is the following
lemma, which is proved in Section~\ref{sec:cores}.
\begin{lemma} 
\label{lem:cores} 
For a random $3$-\sat\ formula of density $\alpha\ge 1$, and for every
$a\in[0,1]$ and $\rho \in[0,1)$,
$$ \lim_{n \rightarrow \infty}\frac{1}{n} \ln\left(\Ex[Y_a]\right) \le 
\sup_r ( f(\alpha, a, r)+ \min\{0, \sup_b h(\alpha, a, r, \rho, b)\}).$$
\end{lemma}
\begin{remark}
This lemma can be strengthened by allowing $\rho$ to depend on $r$;
however, we will not need to use this stronger version.
\end{remark}
\par\medskip

Since every core is a cover we know that $Y_a \le X_a$.  Hence
(in light of Lemma~\ref{lem:covers}) Lemma~\ref{lem:cores} is interesting 
only when $\sup_b h(\alpha, a, r, \rho, b)$ is negative.  In fact, we shall
prove Lemma~\ref{lem:cores} by showing that this supremum bounds 
the logarithm of $1/n$ times the probability that a particular cover of
size $a$ can be extended to a satisfying assignment.  (For a precise
statement, see Lemma \ref{lem:prob} in Section~\ref{sec:cores}.)  An
immediate corollary of Lemma~\ref{lem:cores} is the following.

\begin{corollary}
\label{cor:cores}
If $\alpha\ge 1$, $a\in[0, 1]$ and 
there exists 
$\rho \in [0,1)$ such that 
for every
$r>1$ with $d=\frac{ar}{\alpha} \ln\frac{r}{r-1}\le 1$, 
and for every 
$b\in [0, 1-a]$,
it holds that $f(\alpha, a,r)+h(\alpha, a,r, \rho, b)<0$, then with high
probability random 3-\sat\ formulas of density $\alpha$ do not have
cores of size $an$.
\end{corollary}

The final ingredient in the proof of Theorem~\ref{thm:main} is the
following numerical claim, whose proof we discuss in 
Section \ref{sec:numerical}.

\begin{claim}
\label{claim:numerical}
For every $\alpha \in [4.453, 4.506]$, $a\in [1/(4.506e^2), 1]$, and $r>1$
with $d=\frac{ar}{\alpha} \ln\frac{r}{r-1}\le 1$, 
it holds that either $f(\alpha,a,r)<0$ or for every $b\in [0,1-a]$,
$f(\alpha, a, r)+ h(\alpha, a,r,0.4a+0.7, b)<0$.
\end{claim}
\begin{remark}
In Claim \ref{claim:numerical} we have used the weight
$\rho\equiv\rho(a)=0.4a+0.7$. We
arrived at this particular choice of $\rho$ by first optimizing
numerically for $\rho$ at $100$ values for $a\in [1/4.506, 1]$, and
then fitting a simple function to the values for $\rho$ that
were found. Since Lemma~\ref{lem:cores} holds for every value
of~$\rho$, we may choose any convenient function.
Using a simple analytic function guarantees that $h$ is also
analytic, which makes the numerical analysis of $h$ easier.
\end{remark}
\par\medskip

Finally, combining Claim~\ref{claim:numerical} with Corollary~\ref{cor:cores}
completes the proof of Theorem~\ref{thm:main}.

\subsection{The expected number of covers: Proof of Lemma \ref{lem:covers}}
\label{sec:covers}

Let $s=\lfloor an \rfloor$.  Then we have
\begin{eqnarray*}
\Ex[X_a] &=& \Ex\left[\sum_{\svec\in \{0,1,\ast\}^n} 
\Ind \left[ \mbox{$\svec$ is valid} ~\cap~ (n_\ast(\svec)=n-s) ~\cap~
(n_o(\svec)=0)\right]\right]\\
&=& \sum_{\svec\in \{0,1,\ast\}^n: n_\ast(\svec)=n-s} 
\Pr[\mbox{$\svec$ is valid} ~\cap~ (n_o(\svec)=0)] \\
&=& {n \choose s}~2^s~
\Pr[\mbox{$\svec=(0^s~\ast^{n-s})$ is valid and 
$x_1,\ldots, x_s$ are constrained}].
\end{eqnarray*}

We denote by $P$ the probability that
$\svec=(0^s~\ast^{n-s})$ is valid and all of its assigned variables
are constrained. $P$ is equivalent to the probability of the following event
in an experiment with $m=\alpha n$ balls thrown uniformly and
independently at random into $2^3 {n \choose 3}$ bins. There are 3
kinds of bins:
\begin{enumerate}
\item Bins of type 1 should be empty. These correspond to clauses that
are not allowed:
\begin{itemize}
\item
$(x_{i_1} \vee x_{i_2} \vee x_{i_3})$, with
$i_1, i_2, i_3 \in \{1, 2, \dots, s \}$;
\item
$(x_{i_1} \vee x_{i_2} \vee \bar{x}_j)$, with
$i_1, i_2 \in \{1, 2, \dots, s \}$ and $j>s$;
\item
$(x_{i_1} \vee x_{i_2} \vee {x}_j)$, with
$i_1, i_2 \in \{1, 2, \ldots, s \}$ and $j>s$.
\end{itemize}
The total number of these is 
$$ {s \choose 3}+ 2(n-s){s \choose 2} 
= n^3~\left(\frac{a^3}{6}+ a^2(1-a)\right) + o(n^3).$$
\item Bins of type 2 correspond to constraining clauses: 
$(x_{i_1} \vee x_{i_2} \vee \bar{x}_t)$, with $i_1, i_2, t \in \{1, 2,
\ldots, s \}$. For each variable $x_t$ 
there are $ {s-1 \choose 2} = n^2~\frac{a^2}{2} + o(n^2)$ clauses
that could constrain it and at least one has to be included. Equivalently, 
there has to be at least one ball in one of
those bins for every $x_t$ with $t \in \{1, 2, \ldots, s\}$.
The total number of these clauses is:
$ s {s-1 \choose 2} = n^3~\frac{a^3}{2} + o(n^3).$
\item There are no constraints for the remaining bins, of type~3. Their
total number is
$$2^3 {n \choose 3}- n^3~\left(\frac{a^3}{6}+ a^2(1-a)\right)- n^3~\frac{a^3}{2}+o(n^3) = \frac{n^3}{3}~( 4 - a^2(3-a))+o(n^3).$$
\end{enumerate}

Suppose $m'=dm$ of the
clauses we choose are of type 2, and the remaining $m-m'$ are of type~3.
The probability of this event is
$$ p_{m'}={m \choose m'}\left( \frac{3~a^3}{8}+o(1)
\right)^{m'}  \left( 1- \frac{a^2~(3-a)}{4}+o(1) \right)^{m-m'}.$$
The probability that the $m'$ clauses of type 2 are such that there is
at least one of each kind is the same as the coupon collectors
probability of success, with $s =\lfloor an \rfloor$ different
coupons, and $m'= (d
\alpha) n$ trials. We will use the following general fact from \cite{Chv91}, 
which was previously used in a very similar context in \cite{KKSVZ07}: 
Let $q(cN,
N)$ denote the probability of collecting $N$ coupons within $cN$
trials. If $c<1$, $q(cN, N)=0$. Otherwise, as $N$ goes to infinity
$q(cN, N)$ grows like $g(c)^N$, where
$g(c)=\left(\frac{r_0}{e}\right)^c\frac{1}{r_0-1}$, and $r_0$ is the
solution of $r\ln\left(\frac{r}{r-1}\right)=c$.  More precisely,
$\lim_{N\rightarrow \infty}\frac{1}{N}\ln \left(q(cN, N)\right)=\ln(g(c))$. We have 

\begin{eqnarray*}
P &=& \sum_{m'=0}^m  {m \choose m'}
\left( \frac{3~a^3}{8}+o(1)\right)^{m'}  
\left( 1- \frac{a^2~(6-2a)}{8}+o(1)\right)^{m-m'} ~q(m', s) \\
&\le& m~ \max_{m'} \left\{ 
{m \choose m'} \left( \frac{3~a^3}{8}+o(1)\right)^{m'}  
\left( 1- \frac{a^2~(3-a)}{4}+o(1)\right)^{m-m'} q(m', s). \right\} 
\end{eqnarray*}
Finally,
\begin{eqnarray*}
\Ex[X_a] \le {n \choose s}~2^s m \max_{m'}& 
\left\{ {m \choose m'} \left( \frac{3~a^3}{8}+o(1)\right)^{m'}  
\left( 1- \frac{a^2~(3-a)}{4}+o(1)\right)^{m-m'} q(m', s)
\right\},
\end{eqnarray*}
and hence 
\begin{eqnarray*}
\lim_{n\rightarrow \infty} \frac{1}{n}\ln\left(\Ex[X_a]\right) 
&\le&  \ln\left( \frac{2^a}{a^a(1-a)^{1-a}} \right) 
\\&&
+ \sup_d \left\{ 
\alpha \ln \left(
\frac{ \left(\frac{3~a^3}{8}\right)^d \left( 1- \frac{a^2~(3-a )}{4}\right)^{1-d} }{d^d~(1-d)^{1-d}}\right)
+ a \ln\left(g\left(\frac{d \alpha}{a}\right)\right)
\right\}\\
&=& \sup_d  f(\alpha, a, r) = \sup_r f(\alpha, a, r).
\end{eqnarray*}
This completes the proof of Lemma~\ref{lem:covers}.

\vspace{.1in}

\subsection{The probability that a cover assignment is a core: Proof of
Lemma \ref{lem:cores}}
\label{sec:cores}

For a partial assignment $\svec$, it will be convenient to denote by
$\varphi_\svec(x_{S^\ast(\svec)})$ the formula that is obtained by
substituting the variables that have $0/1$ assignments in $\svec$,
i.e.,  removing from $\varphi$ clauses that are satisfied by at least
one of the assigned variables and removing all remaining appearances
of assigned variables. This is a formula on $n^*(\sigma)$
variables. Notice that if $\svec$ is a valid assignment for~$\varphi$
then the formula $\varphi_\svec$ contains no empty clauses and no
unit clauses. Furthermore all clauses of type~2 (from the previous
section) are removed, because they are satisfied by their constrained
variables. Among the clauses of type~3, there are clauses that are
removed, there are clauses that become two-variable clauses, and there
are clauses that remain untouched. Since there is no simple way to
describe the resulting distribution on formulas, we will keep referring
to the set of all clauses of type~3, even the ones that are removed in
$\varphi_\svec$. When we condition on the fact that $\sigma$ is a
cover and $m'$ of the clauses are of type 2, as in the previous
section, we know that the set of clauses we are interested in are
distributed exactly as a uniform set of $(m-m')$ clauses of type~3.
Thus we can express the expected number of cores as

\begin{eqnarray*}
\Ex[Y_a] 
&=& 
\sum_{\svec \in \{0,1,\ast\}^n:  n_\ast(\svec)=n-s}
\Pr[\mbox{$\svec$ is a cover}] \times 
\Pr[\mbox{$\svec$ is a core }| \mbox{ $\svec$ is a cover }]\\
&=& 
{n \choose s}~2^s~
\Pr[\mbox{$\svec=(0^s~\ast^{n-s})$ is a cover}] \times
\Pr[\mbox{$\svec=(0^s~\ast^{n-s})$ is a core }| \mbox{ $\svec$ is a cover}]\\
&=& 
{n \choose s}~2^s~\times \sum_{m'=s}^m p_{m'} ~q(m', s) \\ &&~~~~
\times 
\Pr[\mbox{$\varphi_\svec(x_{s+1}, \dots, x_n)$ is satisfiable}~| 
~\sigma=(0^s~\ast^{n-s}),
\mbox{ $m-m'$ clauses are of type 3}].
\end{eqnarray*}

We will bound this probability using the Poisson approximation, which
is a standard technique in this area. There are two related random
models.  In the first model, which we call the
\emph{exact model}, $(m-m')$ clauses are chosen uniformly at random
with replacement from all $M=n^3 (4-a^2(3-a))/3+o(n^3)$ clauses of
type 3. In the second model, which we call the \emph{Poisson model},
each of the $M$ clauses is included in the formula with probability
$p=(m-m')/M - 1/(n^2\sqrt{\log{n}}) =
\frac{3\alpha(1-d)}{n^2( 4 - a^2(3-a))}+o\left(\frac{1}{n^2}\right)$.
The expected number of clauses in both models is the same up to a term
$\delta = M/(n^2\sqrt{\log{n}})=\Theta(n/\sqrt{\log{n}})$ which is
sub-linear in $n$.

The Poisson model has been studied before in the context of random
3-\sat. It can be shown, as the example below demonstrates, that
whenever a property holds with high probability in the exact model, it
also holds with high probability in the corresponding Poisson
model. Applying first moment techniques to the Poisson
model is usually easier, because the clauses are independently
chosen; however the bounds obtained are usually weaker.

Next, we relate the probability that $\varphi_\svec$ is satisfiable
under the two models. Let $\Pr_p$ denote probability in the Poisson
model, and $\Pr_e$ denote probability in the exact model.  Let the
random variable $J$ denote the number of \emph{distinct} clauses
included in the formula.  Then
\begin{eqnarray*}
\Pr\nolimits_p[\mbox{$\varphi_\svec$ is satisfiable}]&=& 
\sum_{i=0}^{M}~~~\Pr\nolimits_p [J=i] \times
\Pr[\mbox{$\varphi_\svec$ is satisfiable }|~J=i] 
\\
&\ge&
\sum_{i=0}^{m-m'-\delta}~~~\Pr\nolimits_p [J=i] \times
\Pr[\mbox{$\varphi_\svec$ is satisfiable }|~J=i].
\end{eqnarray*}
Since this conditional probability decreases as $i$ increases, 
and $\Pr_p[J\le \Ex_p[J]]\ge 1/2$, where $\Ex_p[J]=m-m'-\delta$, we have
\begin{eqnarray*}
\Pr\nolimits_p[\mbox{$\varphi_\svec$ is satisfiable}]
&\ge& 
\Pr[\mbox{$\varphi_\svec$ is satisfiable }|~
J = m-m'-\delta ] \times
\sum_{i=0}^{m-m'-\delta} \Pr[J=i] \\
&\ge&
\Pr[\mbox{$\varphi_\svec$ is satisfiable }
|~ J=m-m'-\delta ] \times
\frac{1}{2}.
\end{eqnarray*}
On the other hand, for the exact model
\begin{eqnarray*}
\Pr\nolimits_e[\mbox{$\varphi_\svec$ is satisfiable}]&=& 
\sum_{i=0}^{m-m'}~~~\Pr\nolimits_e[J=i] \times
\Pr[\mbox{$\varphi_\svec$ is satisfiable }|~J=i] 
\\
&\le&
\left( \sum_{i=0}^{m-m'-\delta} \Pr\nolimits_e [J=i] \right) +
\Pr[\mbox{$\varphi_\svec$ is satisfiable }
|~J=m-m'-\delta ] \\
&\le&
\Pr\nolimits_e[\mbox{at least $\delta$ clauses repeated}]+2 
\Pr\nolimits_p[\mbox{$\varphi_\svec$ is satisfiable}] \\
&\le& {m-m' \choose \delta} \left( \frac{m-m'}{M} \right)^\delta +
2 \Pr\nolimits_p[\mbox{$\varphi_\svec$ is satisfiable}] \\
&\le& \left(\frac{(m-m')^2}{M}\right)^\delta +
2 \Pr\nolimits_p[\mbox{$\varphi_\svec$ is satisfiable}] \\
&=& 
\theta(2^{-n\sqrt{\log{n}}})+
2 \Pr\nolimits_p[\mbox{$\varphi_\svec$ is satisfiable}].
\end{eqnarray*}

Thus if the probability of satisfiability in the Poisson model is
bounded above by $c^{-n}$ for some constant $c$, then 
$\lim_{n\rightarrow \infty}
\frac{1}{n}\ln\left(\Pr_e[\mbox{ $\varphi_\svec$ is satisfiable}]\right) 
\le c$.
Therefore, it suffices to get a bound for the Poisson model.  Any
of the first moment techniques for bounding the satisfiability
threshold of 3-\sat~can be adapted to bound the probability that
$\varphi_\svec$ is satisfiable. Of the ones that are technically
applicable (i.e., result in an explicit analytic expression for every
setting of the parameters $s$ and $m'$), we obtain the strongest
result using the novel approach of applying the first moment method to
the distribution on partial assignments defined by the weights in
equation~(\ref{eq:weight}).

\begin{lemma} \label{lem:prob}
In the Poisson model with parameters $n$, $m=\alpha n$, $m'=dm$,
$s=\lfloor an \rfloor$ and for $\svec=(0^s~\ast^{n-s})$, and $r$ such that 
$d=\frac{ar}{\alpha}\ln\frac{r}{r-1}$,
\begin{eqnarray*}
\lim_{n\rightarrow \infty} \frac{1}{n} 
\ln\left(\Pr\nolimits_p[\mbox{ $\varphi_\svec(x_{s+1}, \dots, x_n)$ 
is satisfiable}]\right)
\le \inf_{\rho \in [0,1)} ~ \sup_{b \in [0,1-a]} h(\alpha, a, r, \rho, b) .
\end{eqnarray*}
\end{lemma}
\begin{proof}
We will apply Theorem \ref{thm:sum} to the formula
$\varphi_\svec(x_{s+1}, \dots, x_n)$. The theorem implies that if
$\varphi_\svec$ has a satisfying assignment then $\sum_{\tau \in V}
W(\tau)\ge 1$ where $W(\tau)=\rho^{n_\ast(\tau)}(1-\rho)^{n_o(\tau)}$
and the sum is over the set $V$ of all partial assignments $\tau \in
\{0,1,\ast\}^{n-s}$ that are valid for $\varphi_\svec$.  
Thus the probability of satisfiability is bounded from above by
the expected value of $\sum_{\tau \in V} W(\tau)$.  This holds for any value
of $\rho \in [0,1]$. We bound this expectation.

For any $t \in \{0, 1, \ldots, n-s\}$ let $Z_t$ denote the sum of the
weights of valid assignments $\tvec$ for $\varphi_\svec(x_{s+1},
\dots, x_n)$ such that $n_\ast(\tvec) = n-s-t$. Then $$\Ex[Z] =
\sum_{t=0}^{n-s}~\Ex[Z_{t}] ~\le~ n~\max_{t \in\{0, 1, \dots, n-s\}}
\Ex[Z_{t}].$$ 
\begin{eqnarray*}
\Ex[Z_{t}] &=& \sum_{u=0}^t \rho^{n-s-t} (1-\rho)^u 
\sum_{\tvec\in \{0,1,\ast\}^{n-s}} 
\Pr\nolimits_p 
\left[ \mbox{$\tvec$ is valid} ~\cap~ (n_\ast(\tvec)=n-s-t) ~\cap~
(n_o(\tvec)=u)\right] \\
&=& \rho^{n-s-t}~\sum_{u=0}^{t}
(1-\rho)^u~\sum_{\tvec\in \{0,1,\ast\}^{n-s}~:~n_\ast(\tvec)=n-s-t} 
\Pr\nolimits_p[\mbox{$\tvec$ is valid} \cap (n_o(\tvec)=u)] \\
&=& \rho^{n-s-t}~\sum_{u=0}^{t} (1-\rho)^u~{n-s \choose t}~{t \choose u}~
2^{t}~ \\
&& \qquad \qquad \qquad
\Pr\nolimits_p
[\mbox{$\tvec=(0^{t}~\ast^{n-s-t})$ is valid and $x_{s+1}, \dots, x_{s+u}$  
are unconstrained}] 
\end{eqnarray*}

Next we derive the probability that the assignment $(x_{s+1}, x_{s+2},
\ldots, x_n)=(0^t~\ast^{n-s-t})$ is valid and the first $u$ variables
are unconstrained. 
Recall that $\varphi_\sigma$ is obtained from
$\varphi$ by setting its first $s$ variables according to
$\sigma$. Only clauses of type 3 influence $\varphi_\svec$ and
according to the Poisson model, each of them is included independently
with probability $p$. The probability that the assignment $(x_{s+1}, x_{s+2},
\ldots, x_n)=(0^t~\ast^{n-s-t})$ is valid and the first $u$ variables
are unconstrained is equivalent to the probability that among the
included clauses of type 3, there are no clauses of the following
kinds:
\begin{itemize}
\item
$(x_{i_1} \vee x_{i_2} \vee x_{i_3})$, with 
$i_1  \in \{1, 2, \dots, s+t \}$, 
$i_2, i_3 \in \{s+1, s+2, \dots, s+t \}$, 
\item
$(x_{i_1} \vee x_{i_2} \vee \bar{x}_j)$, with
$i_1 \in \{1, 2, \dots, s+t\}$, 
$i_2 \in \{s+1, s+2, \dots, s+t \}$ and $j>s+t$, 
\item
$(x_{i_1} \vee x_{i_2} \vee {x}_j)$, with
$i_1 \in \{1, 2, \dots, s+t\}$, 
$i_2 \in \{s+1, s+2, \dots, s+t \}$ and $j>s+t$, 
\item 
$(x_{i_1} \vee x_{i_2} \vee \bar{x}_j)$, with
$i_1 \in \{1, 2, \dots, s+t\}$, 
$i_2 \in \{s+1, s+2, \dots, s+t \}$ and $j \in \{s+1, \dots, s+u\}$, 
\end{itemize}
and for every $j \in \{s+u+1, \dots, s+t\}$, the formula contains a
clause $(x_{i_1} \vee x_{i_2} \vee \bar{x}_j)$, with $i_1 \in \{1, 2,
\dots, s+t\}$, $i_2 \in \{s+1, \ldots, s+t \}$. 

In the Poisson model all clauses are independent, so it is easy to put
these events together to obtain:
\begin{eqnarray*}
Q&=& \Pr\nolimits_p[\mbox{$\tvec=(0^{t}~\ast^{n-s-t})$ is valid and $x_{s+1},
\dots, x_{s+u}$ are unconstrained}] \\ &=& 
(1-p)^{ \left({t \choose 3}+ s{t \choose 2}\right) 
~+~ 2(n-s-t)\left({t \choose 2}
~+~ st\right) ~+~ u\left( {t \choose 2}+st\right)}
\times \left(1-(1-p)^{{t \choose 2}+st}\right)^{t-u} 
\end{eqnarray*}
The expression for the expectation can be simplified:
\begin{eqnarray*}
\Ex[Z_t] 
&=& \rho^{n-s-t}~{n-s \choose t}~2^t 
\sum_{u=0}^{t} (1-\rho)^u~{t \choose u}~Q\\
&=& \rho^{n-s-t}~{n-s \choose t}~2^t~
(1-p)^{{t \choose 3} ~+~ s{t \choose 2} 
~+~ 2(n-s-t)\left({t \choose 2}~+~ st\right)}\\
&& \qquad \qquad \times \sum_{u=0}^{t} {t \choose u}~(1-\rho)^u~
(1-p)^{u\left( {t \choose 2}+st\right)}
\times \left(1-(1-p)^{{t \choose 2}+st}\right)^{t-u} \\
&=& \rho^{n-s-t}~{n-s \choose t}~2^t~
(1-p)^{{t \choose 3} ~+~ s{t \choose 2} 
~+~ 2(n-s-t)\left({t \choose 2}~+~ st\right)}\\
&& \qquad \qquad \times \left( (1-\rho)(1-p)^{{t \choose 2}+st}+
\left(1-(1-p)^{{t \choose 2}+st}\right) \right)^t \\
&=& \rho^{n-s-t}~{n-s \choose t}~2^t~
(1-p)^{{t \choose 2}(6n-5t-3s-2)/3 + 2(n-s-t)st}
\left(1-\rho(1-p)^{{t \choose 2}+st}\right)^t 
\end{eqnarray*}

Let $t=bn$, and recall that $s=\lfloor an \rfloor$, 
$p=\frac{3\alpha(1-d)}{n^2( 4 - a^2(3-a))}+o\left(\frac{1}{n^2}\right)$. Then 

\begin{eqnarray*}
\lim_{n\rightarrow \infty} \frac{1}{n} \ln\left(\Ex[Z_t]\right) &=& 
\ln\left(\frac{(1-a)^{1-a} ~2^b ~\rho^{1-a-b}}{b^b ~(1-a-b)^{1-a-b}}\right)
 -~ \frac{\alpha~(1-d)~b~(b(6-5b-3a) + 12(1-a-b)a)}{2(4 - a^2(3-a))}\\
&& +~ b \ln\left(1-\rho  
e^{ -\frac{3\alpha(1-d)b(b+2a)}{2(4-a^2(3-a))}}\right) \\
&=& h(\alpha, a, r, \rho, b).
\end{eqnarray*}

Finally, 
$$\lim_{n\rightarrow \infty} \frac{1}{n} 
\ln \left(
\Pr\nolimits_p[\mbox{ $\varphi_\svec(x_{s+1}, \dots, x_n)$ is satisfiable}] 
\right)
\le 
\lim_{n\rightarrow \infty} \frac{1}{n} 
\ln \left(\Ex[Z]\right) \le 
\sup_{b \in [0,1-a]} h(\alpha, a, r, \rho, b),$$
which is the statement of Lemma~\ref{lem:prob}.
\end{proof}
\par\medskip

Substituting the bound of Lemma \ref{lem:prob} into the expression for the
expectation yields
$$\lim_{n\rightarrow \infty} \frac{1}{n}\ln \left(\Ex[Y_a]\right) \le 
\sup_r (f(\alpha, a, r) + \min\{0, \sup_b h(\alpha, a, r, \rho(a,r), b)\}),$$
completing the proof of Lemma~\ref{lem:cores}.

\subsection{Numerical analysis: Steps to the proof of Claim \ref{claim:numerical}}
\label{sec:numerical}

What remains is to verify the numerical claim that for every $\alpha
\in [4.453, 4.506]$, $a\in [1/(4.506e^2), 1]$, and $r>1$ with
$d=\frac{ar}{\alpha} \ln\frac{r}{r-1}\le 1$, it holds that either
$f(\alpha,a,r)<0$ or for every $b\in [0,1-a]$, $f(\alpha, a, r)+
h(\alpha, a,r,0.4a+0.7, b)<0$. We outline the steps towards a 
proof.

There are two stages. First, we identify the $(a,r)$ pairs
such that for every $\alpha \in [4.453, 4.506]$ it holds that
$f(\alpha, a, r)<0$. Second, for the remaining range of values of
$(a,r)$ we show that for every $b\in [0,1-a]$, $f(\alpha, a, r)+
h(\alpha, a,r,0.4a+0.7, b)<0$.

The derivative of $f$ with respect to $r$ is
$$\frac{\partial f}{\partial r}=
a \left( \ln\left(1+\frac{1}{r-1}\right)-\frac{1}{r-1}\right)
\times \ln\left(\frac{3\left(\alpha-ar\ln\left(\frac{r}{r-1}\right)\right)}{2\left(\frac{4}{a^2}-3+a\right)\ln\left(\frac{r}{r-1}\right)}\right).$$
Since the first factor is always negative the derivative is 0 when 
\begin{equation}
\label{eq:rmax}
3\alpha =
\ln\left(\frac{r}{r-1}\right)\left(\frac{8}{a^2}-6+2a+3ar\right).
\end{equation}
This equation has a unique root for every $a\in [1/(4.506e^2),1],
\alpha\in [4.453, 4.506]$ (the derivative of the right-hand side is
negative). Therefore we can conclude that for every $a$ and $\alpha$
in the above range, $f$ is at first increasing with $r$ then
decreasing. Its maximum is achieved at the root of
equation~(\ref{eq:rmax}). Figure \ref{fig:posr} shows the values for
$r$ where $f$ is maximized, with respect to $a$, when $\alpha=4.453$.

\begin{figure}
\begin{center}
\psfig{file=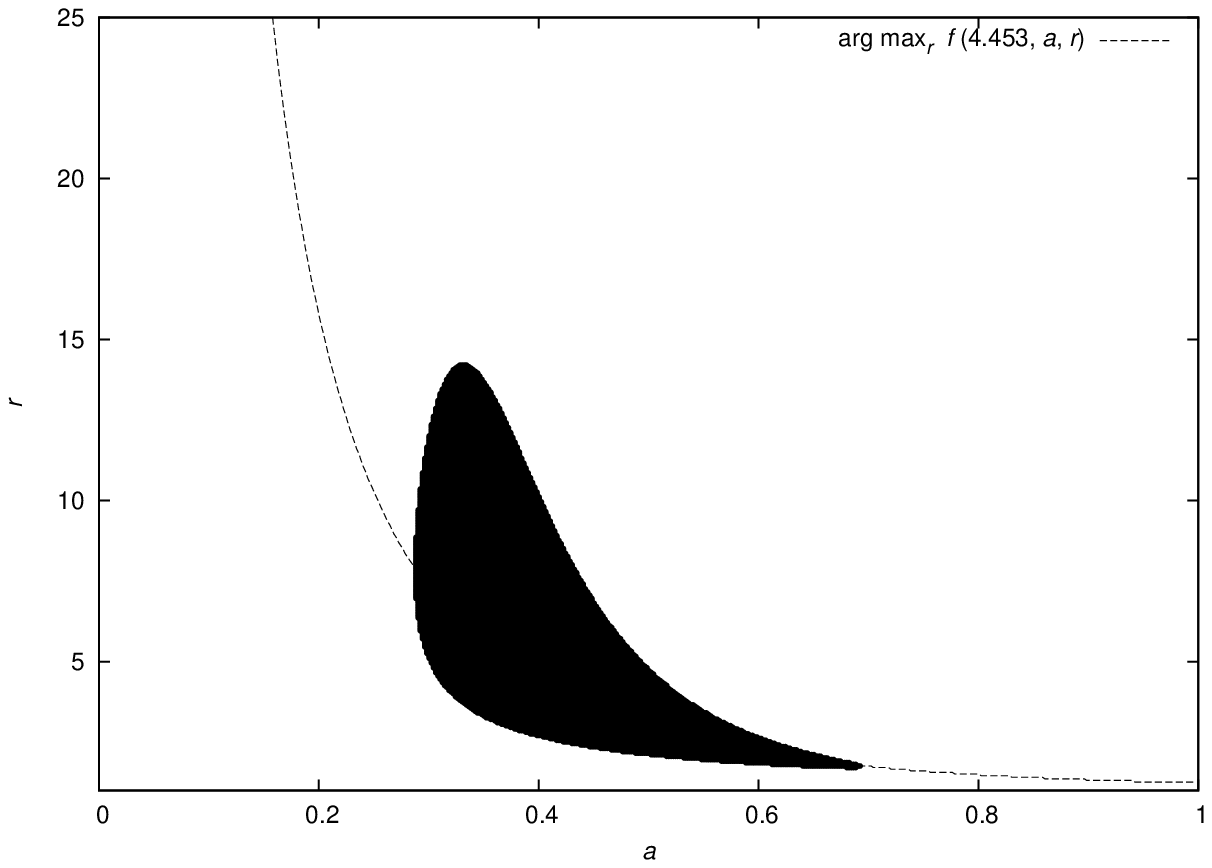,height=3in} 
\end{center}
\caption[]{The pairs $(a,r)$ such that $f(4.453, a, r)>-0.0001$, and the 
value of $r$ that maximizes $f(4.453, a, r)$ for every $a$.}
\label{fig:posr}
\end{figure}

Since the root of~(\ref{eq:rmax}) is monotone decreasing with respect
to $a$ and $\alpha$, we can find values $r_1$ and $r_2$ such that for
every $a$ and $\alpha$ in the range, if $r<r_1$ then $\partial
f/\partial r>0$, and if $r>r_2$ then $\partial f/\partial r<0$. Values
satisfying this condition are $r_1=1.2$ and $r_2=670$. Thus if we show
that $f(\alpha, a, r)$ is negative for $r=r_1$ and $r=r_2$ then it is
negative for any $r$ outside the range $(1.2, 670)$.

We will use the shorthand notation $q:=r\ln\frac{r}{r-1}$. 
For $r\in (1.2, 670)$ we have $q\in (1.0007, 2.16)$. The whole range
satisfies the condition that $d=\frac{a r}{\alpha}r
\ln\frac{r}{r-1}\le1$.

Next we take care of the boundary region with respect to $a$. Notice
that the derivative of $f$ is
negative for $a$ large enough (close to 1), because of the entropy
term in $f$.  The derivative of $f$ with respect to $a$ is
\begin{eqnarray*}
\frac{\partial f}{\partial a} &=& \ln(2) - \ln\frac{a}{1-a}- 
q\ln\frac{3a^2(\alpha-aq)}{2q(4-a^2(3-a))}+3q
-\frac{3a(\alpha-aq)(2-a)}{4-a^2(3-a)}
+q\ln(r)-\ln(r-1).
\end{eqnarray*}

The following observations are helpful for bounding this derivative:
\begin{itemize}
\item $3q+q\ln(r)-\ln(r-1)$ is maximized in the interval $r\in[1.2, 670]$ at 
$r=1.2$.
\item $a^2(\alpha-aq)/(4-a^2(3-a))$ is an increasing function of $a$.
\item $a(\alpha-aq)(2-a)/(4-a^2(3-a))$ is an increasing function of $a$.
\end{itemize}

For $a>0.999$, using the above facts, the derivative is negative. Thus if
we show that $f$ is negative for $a=0.999$, $r\in (1.2, 670)$, $\alpha
\in [4.453, 4.506]$ then $f$ is negative also for every $a>0.999$.
 
We are left with the region $a\in [1/(4.506e^2), 0.999]$, $r\in (1.2,
670)$. In this region all the derivatives of $f$ can be bounded,
and a sufficiently fine grid can be chosen over which to evaluate $f$
in order to identify the grid sections where the function can take
positive values. In particular, the derivative with respect to $a$ is
at most $28.2$, and with respect to $\alpha$ it is at most~$1$. 
Furthermore, since we know that for every $a$ and $\alpha$,
$f$ is maximized as a function of~$r$ at the root of equation~(\ref{eq:rmax}), 
one can find this maximum and the range of~$r$ where
$f(\alpha, a, r)$ is positive using binary search. The points with
$f(4.453, a, r)>-0.0001$ are depicted in Figure~\ref{fig:posr}.

In the second stage we need to analyze $f(\alpha, a, r)+h(\alpha, a,
r, 0.4a+0.7, b)$ for the remaining region of values for $(a,r)$.

First, notice that we can take care of the boundary regions with
respect to $b$ by taking advantage of the entropy term, which has very
large slope for $b$ close to 0, and very steep negative slope for $b$ close to
$1-a$. Specifically, the derivative with respect to $b$ is
\begin{eqnarray*}
\frac{\partial h}{\partial b} &=& \ln(2) - \ln(\rho)-
\ln\left(\frac{b}{1-a-b}\right) 
- \alpha (1-d) \frac{12b-15b^2-30ab+12a-12a^2}{2(4-a^2(3-a))} \\
&&{} + \ln(1-\rho e^{-A})+
3\alpha(1-d)b~ 
\frac{\rho e^{-A}(a+b)}{(1-\rho e^{-A})(4-a^2(3-a))}
\end{eqnarray*}
where $A= \frac{3\alpha(1-d)b(b+2a)}{2(4-a^2(3-a))}$.

Using the bounds on all parameters: $\alpha \in [4.453, 4.506]$, $a\in
(0.28, 0.7)$, $r\in (1.5,14.3)$, and the ones that follow from those:
$d\in (0.06, 0.26)$, $A\in [0, 1.59]$, $\rho \in (0.02, 0.21)$, we can
conclude that $$ -3.04 -\ln\left(\frac{b}{1-a-b}\right) 
< \frac{\partial h}{\partial b}
< 5.38-\ln\left(\frac{b}{1-a-b}\right).$$

The lower bound can be made positive by setting $b/(1-a-b)\le 0.04$ and
the upper bound can be made negative by setting
$b/(1-a-b)\ge220$. Therefore it suffices to show that $h$ is negative for
$b$ in the range $[0.01,0.996(1-a)]$. 

Again, using the bounds for the parameters, all the derivatives of
$f+h$ can be bounded, and a sufficiently fine grid can be chosen over
which to evaluate $f+h$. We did not perform the evaluation on a grid
that is as fine as is required for the rigorous proof because, based on
the current bounds on the derivatives, we would need to evaluate the
function at more than $10^{10}$ points. (As we discussed in the
introduction, we feel that the computational effort is not yet
justified.)

For illustration, Figure~\ref{fig:expcovers} shows the estimated
values of $\sup_r f(\alpha, a, r)$ and of $\sup_{r,b}(f(\alpha,a,r)$
$+$ $h(\alpha, a, r, \rho(a), b))$ for $\alpha=4.453$, $4.470$,
$4.490$, and $4.506$. The maximum is taken over evaluations at a grid
of step size $0.001$ for the parameters $r$ and $b$.  We estimated
with better precision the location of the maximum of $f+h$ by
evaluating the function at a fine grid in the region where the
evaluations on the coarse grid give the largest values. The maximum
found in this way is at $\alpha=4.453$, $a=0.62566$, $b=0.03568$ and
$r=2.00134$ ($d=0.19473$). The value of $f+h$ at this point is
$-0.000058$.

\begin{figure}
\begin{center}
\psfig{file=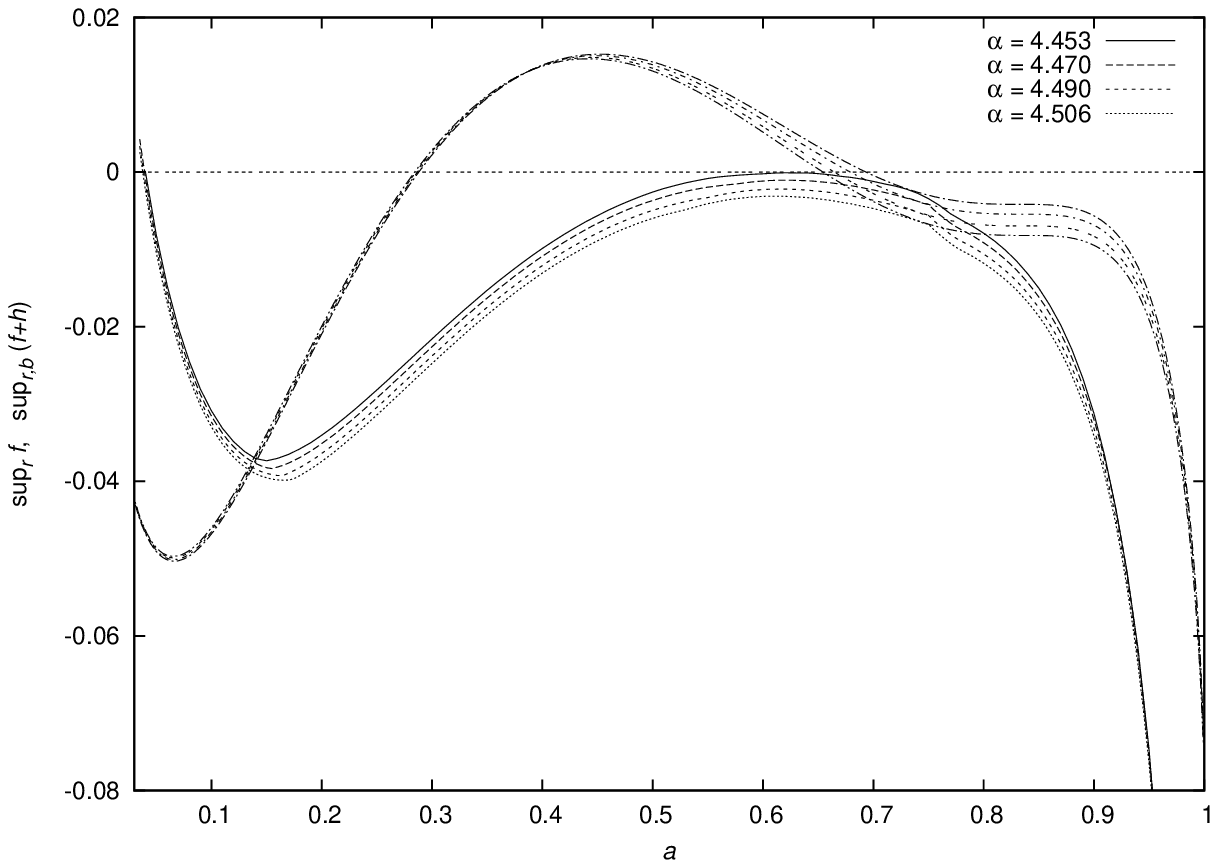,height=3in} \\
\end{center}
\caption[]{
The values of
$\sup_r f(\alpha, a, r)$ and 
$\sup_{r,b} (f(\alpha, a,r)+ h(\alpha, a, r, 0.4 a+0.7, b))$
for  $\alpha=4.453$, $4.470$, $4.490$, and $4.506$.}
\label{fig:expcovers}
\end{figure}

\comment{
\begin {figure}
      \begin{center}
        \input{exactcoredoubleb.tex}
      \end{center}
\caption{
The estimated value of 
$\sup_{r,b} (f(\alpha, a,r)+ h(\alpha, a, r, 0.4 a+0.7, b))$
for  $\alpha=4.453$, $4.47$, $4.49$, and $4.506$.}
\label{fig:expcovers}
\end {figure}
}

\section{Concluding remarks}

As mentioned in the introduction, all of the known rigorous upper
bounds for the satisfiability threshold of 3-\sat\ are based on the
first moment method
\cite{KMPS95,DB97,KKKS98,JSV00,KKSVZ07,DBM00}; the corresponding 
upper bounds in this sequence of results are: 4.758, 4.64, 4.601,
4.596, 4.571, 4.506. The general method is to consider a random
variable $Z$ that is equal to the number of satisfying assignments of
a particular kind. These satisfying assignments are such that at least
one exists if the formula is satisfiable.  Showing that above a certain
density $\Ex[Z] \rightarrow 0$ thus implies (by Markov's inequality) that
$\Pr[Z=0]\to 1$, and consequently the probability that the formula
is unsatisfiable also goes to~1. For example, \cite{DB97}~takes $Z$ to
be the number of {\em negatively prime solutions}, i.e.,
solutions for which every variable assigned~1 is constrained.

Here, we have used the same idea but with $Z$ taken to be the total weight of
partial assignments under the weight function~(\ref{eq:weight}) inspired
by Survey Propagation.  Given the dramatic success of Survey Propagation
for random 3-\sat, it seems plausible that this approach can potentially
yield rather tight upper bounds on the threshold.  We were able to achieve
only partial progress in this direction, but we are quite hopeful that
extensions of our approach could lead to further progress.

\comment{
One natural extension would be to consider the total
weight of partial assignments that can be reached from a satisfying
assignment, or the total weight of partial assignments that can be
reached from positively prime solutions, etc. The intuition is that
this idea may allow us to avoid the overcounting due to formulas with
many solutions in a single cluster, by merging a large part of the
weight of such solutions. 
} 

One natural extension would be to use a different weight $\rho$
for each variable, depending for example on the number of positive and
negative occurrences of the variable. The corresponding generalization of
Theorem~\ref{thm:sum} is proved in
\cite{AM07}. It is quite possible that the value of $\Ex[Z]$ in this case is
significantly smaller.

\section{Acknowledgments}

We would like to thank the anonymous referees for their detailed and very
useful suggestions.

\bibliographystyle{plain}
\bibliography{firstmoment}

\end{document}